\documentclass[11pt]{article}
\usepackage{subcaption}
\usepackage{amssymb}
\usepackage{amsthm, amsmath}
\usepackage{tikz}
\usepackage{hyperref}
\usepackage[margin = 1in]{geometry}
\usepackage{algorithm2e}
\usepackage{tikz}

\title{Coloring Down: $3/2$-approximation for special cases of the weighted tree augmentation problem\footnote{We would like to acknowledge Robert Carr and L\'aszl\'o V\'egh for their numerous discussions on this problem. In addition, we would like to thank Sinziana Munteanu for her developments on the structure of the trees. }}
\author{Jennifer Iglesias \thanks{This material is based upon work supported by the National Science Foundation Graduate Research Fellowship Program under Grant No. 2013170941.} \and R. Ravi\thanks{This material is based upon research supported in part by the U. S. National Science Foundation under award number CCF-1527032.}
}

\newcommand{\newthmwithin}[3]{\newtheorem{#1q}{#2}[#3]
                        \newenvironment{#1}{\begin{#1q}\sf}{\end{#1q}}}

\newcommand{\newthm}[3]{\newtheorem{#1q}[#2q]{#3}
                        \newenvironment{#1}{\begin{#1q}\sf}{\end{#1q}}}
																								
\newthmwithin{theorem}{Theorem}{section}
\newthm{definition}{theorem}{Definition}
\newthm{lemma}{theorem}{Lemma}
\newthm{conjecture}{theorem}{Conjecture}
\newthm{corollary}{theorem}{Corollary}

\newcounter{note}[section]

\begin{document}
\maketitle

\begin{abstract}

In this paper, we investigate the weighted tree augmentation problem (TAP), where the goal is to augment a tree with a minimum cost set of edges such that the graph becomes two edge connected. First we show that in weighted TAP, we can restrict our attention to trees which are binary and where all the non-tree edges go between two leaves of the tree. We then give two different top-down coloring algorithms. Both of our algorithms differ from known techniques for obtaining a $\frac32$-approximation in unweighted TAP and current attempts to reach a $\frac32$-approximation for weighted TAP.

The first algorithm we describe always gives a 2-approximation starting from any feasible fractional solution to the natural tree cut covering LP. When the structure of the fractional solution is such that all the edges with non-zero weight are at least $\alpha$, then this algorithm achieves a  $\frac{2}{1+\alpha}$-approximation. We propose a new conjecture on extreme points of LP relaxations for the problem, which if true, will lead to a potentially constructive proof of an integrality gap of at most $\frac32$ for weighted TAP.
In the second algorithm, we introduce simple extra valid constraints to the tree edge covering LP. In this algorithm, we focus on deficient edges, edges which get covered to an extent less than $\frac43$ in the fractional solution. We show that in the support of extreme points for this LP, deficient edges occurs in node-disjoint paths in the tree.
When the number of such paths is at most two, we
give a top-down coloring algorithm which decomposes $\frac32$ times the fractional solution into a convex combination of integer solutions.
We believe our algorithms will be useful in eventually resolving the integrality gap of linear programming formulations for TAP.

We also investigate a variant of TAP where each edge in the solution must be covered by a cycle of length three (triangle). We give a $\Theta(\log n)$-approximation algorithm for this problem in the weighted case and a $4$-approximation in the unweighted case.
\end{abstract}


\newpage
\section{Introduction}
We consider the {\em weighted tree augmentation problem (TAP)}: Given an undirected graph $G=(V,E)$ with non-negative weights $c$ on the edges, and a spanning tree $T$, find a minimum cost subset of edges $A \subseteq E(G) \setminus E(T)$ such that $(V, E(T) \cup A)$ is two-edge-connected. We will call the elements of $E(T)$ as (tree) edges and those of $E(G) \setminus E(T)$ as {\em links} for convenience. A graph is \emph{two-edge connected} if the removal of any edge does not disconnect the graph, i.e., it does not have any cut edges.
Since cut edges are also sometimes called bridges, this problem has also been called {\em bridge connectivity augmentation} in prior work~\cite{frederickson1981approximation}.

While TAP is well studied in both the weighted and unweighted case~\cite{frederickson1981approximation, khuller1993approximation, ravithesis, cohen2013, CheriyanG15a, 32KortsarzNutov, adjiashvili2016improved, fiorini2017frac}, it is NP-hard even when the tree has diameter $4$~\cite{frederickson1981approximation} or when the set of available links form a single cycle on the leaves of the tree $T$~\cite{cheriyan1999}. Weighted TAP remains one of the simplest network design problems without a better than $2$-approximation. TAP can also be viewed as a covering problem. The cuts in a tree which have a single edge crossing them are exactly the cuts that must be covered.

A link $\ell$ is said to \emph{cover} an edge $e$ if the unique cycle of $\ell+T$ contains $e$. Here we use $\delta(e)$ for a tree edge $e$ to denote the set of links which cover $e$.
The natural covering linear programming relaxation for the problem, EDGE-LP, is a special instance of a set covering problem with one requirement (element) corresponding to each cut edge in the tree (Since the tree edges define shores that form a laminar family, this is also equivalent to a laminar cover problem~\cite{cheriyan1999}).
\begin{align}
\min \sum_{\ell \in E} &c_\ell x_\ell& \nonumber \\
\label{eq:edge} x(\delta(e))&\geq 1 \quad &\forall e\in T \\
\label{eq:nonneg} x_\ell&\geq 0 \quad &\forall \ell\in E
\end{align}

Fredrickson and J\'aj\'a showed that the integrality gap for EDGE-LP can not exceed $2$~\cite{frederickson1981approximation} and also studied the related problem of augmenting the tree to be two-node-connected (biconnectivity versus bridge-connectivity augmentation)~\cite{fredrickson1982relationship}. Cheriyan, Jord\'an, and Ravi, who studied half-integral solutions to EDGE-LP and proved an integrality gap of $\frac43$ for such solutions, also conjectured that the overall integrality gap of EDGE-LP was at most $\frac43$~\cite{cheriyan1999}. However, Cheriyan et al.~\cite{cheriyan2008integrality} demonstrated an instance for which the integrality gap of EDGE-LP is at least $3/2$.

We study the integrality gap of the EDGE-LP and its generalizations in this work.
We first show that without loss of generality, we can focus our attention on binary trees where every node has degree 1 or 3 (and every link goes between a pair of leaves). By focusing on the internal nodes of degree 3, we can add a simple valid constraint. In particular,  at any node of degree 3, since no link can cover all three edges which meet at this node, the total number of (integral) links which must cover its neighbors is at least 2.
This gives one additional constraint per internal node that we can add to the EDGE-LP. The resulting LP, called the NODE-LP follows where we use $\delta_T(v)$ for a node $v$ to refer to its three incident edges in the tree $T$.

\begin{align}
\min \sum_{\ell \in E} &c_\ell x_\ell& \nonumber \\
 x(\delta(e))&\geq 1 \quad &\forall e\in T \nonumber \\
\label{eq:node}x(\delta(e_1)\cup\delta(e_2)\cup \delta(e_3)) &\geq 2 \quad & \forall v\in T \text{ and } \delta_T(v) = \{e_1,e_2,e_3\}\\
 x_\ell&\geq 0 \quad &\forall \ell\in E \nonumber
\end{align}

Fiorini et. al extended node constraints for all classes of odd subsets of tree edges as $\{0,\frac12\}$-Chv{\'a}tal-Gomory cuts of EDGE-LP to obtain new constraints on all odd sets of edges~\cite{fiorini2017frac}. We call their extended linear program the ODD-LP. Since we will show that we can assume the tree is binary, every node has odd degree (1 or 3) in the input tree, so if $S\subseteq V$ is odd, then it follows that  $\delta(S)\cap T$ is also odd. Using this observation, we can write the ODD-LP as follows.
Recall that $\delta(S)$ for $S \subset V$ is the set of all edges and links with exactly one endpoint in $S$.
\begin{align}
\min \sum_{\ell \in E} &c_\ell x_\ell& \nonumber \\
\label{eq:odd} x(\delta(S))  + \sum_{e\in \delta(S)\cap T} x(\delta(e)) &\geq |\delta(S)\cap T|+1 \quad &\forall S \subseteq V, |S| \text{ odd} \\
 x_\ell&\geq 0 \quad &\forall \ell\in E \nonumber
\end{align}

In Appendix~\ref{app:odd}, we provide a simple independent proof of validity of these odd set constraints due to Robert Carr.

In addition to the standard version, we also study the problem of 3TAP in which every tree edge in the final solution must be in a cycle of length 3 (instead of every tree edge being in a cycle of any length). This is a natural variant of TAP. While TAP models increasing the resilience of a tree network, 3TAP requires local resilience: i.e., in case of any edge failure, the overhead of implementing a rerouting protocol is not too high (3TAP solutions only need the identity of the midpoint of the alternate 2-path for every edge in the solution). 

\subsection{Related Work}

Weighted TAP has several 2-approximation algorithms.
The earliest proof of this result used methods that were tailored for this problem: Frederickson and J\'aJ\'a~\cite{frederickson1981approximation} convert the problem into one of finding a minimum weight arborescence in an appropriate directed graph: First, they root the given tree at an arbitrary node and direct it outwards; Links that go from a node to an ancestor are directed upward in the tree, while cross links are replaced by two links of the same weight going from each endpoint to their least common ancestor in the tree. After given the original tree edges directed downward weight zero, their method finds a minimum weight in-arborescence pointing to the root, which they argue is of cost at most twice the optimal weighted TAP solution for this instance (coming from the duplication of cross links). Khuller and Thurimella improved the runtime of this algorithm~\cite{khuller1993approximation}. It is also worth noting that the directed instance when viewed as an undirected instance of TAP consists of all links going top-down in the tree (since cross links are replaced with two such links from their ends to their lca). The EDGE-LP for all links going top-down in a tree is totally unimodular (see, e.g., Section 2 of ~\cite{golovin2006approximating}). Hence this version can be solved to optimality (providing an alternate to the use of the in-arborescence algorithm).
Later,  other $2$-approximation algorithms have been devised for weighted TAP using other techniques such as the primal-dual method~\cite{ravithesis} and iterative rounding~\cite{jain2001factor}.

Special cases of weighted TAP has also been investigated.
Cheriyan, Jord\'an and Ravi~\cite{cheriyan1999} developed a $\frac43$-approximation for TAP when the optimal fractional solution is half-integral.
Another special case of weighted TAP is when the tree has bounded depth. In this special case, Cohen and Nutov showed there exists a $(1+\ln 2)$-approximation~\cite{cohen2013}.
Recently, Adjiashvili~\cite{adjiashvili2016improved} showed a 1.96-approximation for another special case of weighted TAP where all link weights are between $1$ and some constant $M$ by using a bundling type linear program.
Building off this work,  Fiorini et. al~\cite{fiorini2017frac} generalized the constraints from~\cite{kortsarz2015} and combined them with the bundle constraints from~\cite{adjiashvili2016improved} to propose the ODD-LP we described above and achieved a $\frac32+\epsilon$ approximation for the same special case (when all the costs are between $1$ and some constant $M$). Another recent paper by Nutov takes a subset of Adjiashvili's constraints and achieves a $\frac{12}{7}+\epsilon$ approximation when all the costs are between $1$ and some constant $M$~\cite{nutov2017}. All of these techniques rely heavily on the bundle constraints that are focused on link weights being in a bounded range; hence they do not seem to be generalizable to the case of arbitrary weights. We believe the general problem requires a more polyhedral approach of the type we investigate.

Numerous papers attempted to reach a target $\frac32$-approximation in the unweighted case of TAP when all links have the same weight. One paper by Kortsarz and Nutov~\cite{kortsarz2015} presents a new linear program with a 1.75-approximation for the unweighted case, in the hope that this linear program could help break the 2-approximation barrier for the weighted case. This LP used properties of an optimal solution for the unweighted case to add multiple new constraints; In retrospect, these additional constraints are all included in the ODD-LP.
Two papers achieved a $\frac32$-approximation for the unweighted case with very different approaches; one paper by Kortsarz and Nutov relies on a unique token giving argument~\cite{32KortsarzNutov}. The other paper by Cheriyan and Gao uses semi-definite programming~\cite{CheriyanG15a, CheriyanG15b} to arrive at an initial fractional solution for which this integrality gap is proved. While both of these approaches are very different, they still heavily rely on the fact that all the links have the same weight.

\subsection{Our Results}
Our results gives new approaches to determine the integrality gap of weighted TAP: our methods provide constructive proofs of convex decompositions of given fractional solutions appropriately scaled into integer solutions.
\begin{enumerate}
\item We show that any instance of weighted TAP can be reduced to equivalent instances where the underlying tree is binary and all the links have their endpoints at leaves (Theorem~\ref{thm:structure} in Section~\ref{sec:structure}). The simpler structure of input instances help us in several of our proofs and may also be key in future approaches in settling the integrality gap of weighted TAP.
\item We give a simple new top-down coloring algorithm that gives a constructive proof of the integrality gap of 2 for EDGE-LP by providing a convex decomposition.
    Furthermore, if the minimum non-zero value in the solution for any link is $\alpha$ then we can achieve an improved $\frac{2}{1+\alpha}$-approximation (Theorem~\ref{thm:largelinksTwo} in Section~\ref{sec:large}). This result generalizes the result of Cheriyan et al.~\cite{cheriyan1999}which we can recover by setting $\alpha=\frac12$. Even more interestingly, this provides a new $\frac32$-approximation when all nonzero values in the solution are at least $\frac13$.

\item We provide a new conjecture on the ODD-LP (Conjecture~\ref{con:onethird}) that says that every vertex solution to this LP has all large nonzero entries (greater than $\frac13$) or there is a single very large valued entry (at least $\frac23$). In the former case, we can use the previous theorem to get a $\frac32$-approximation while in the latter, we can apply one step of iterative rounding~\cite{lrs}, and reapply the conjecture to prove a $\frac32$-approximation.

\item We provide a $\frac32$-approximation for weighted TAP based on fractional solutions to NODE-LP with a particular structure. Let a \emph{deficient} edge be an edge which gets covered to the extent less than $\frac{4}{3}$ by this fractional solution. In Section~\ref{sec:deficient}, we show that if the deficient edges for the NODE-LP form at most two paths in the tree, then we can extend our coloring construction to give a $\frac32$-approximation.
    
\item Even though we provide improved approximations for specially structured extreme points of NODE-LP, we can show that such constraints do not strengthen EDGE-LP. In particular, in Section~\ref{sec:obs}, we show how to transform any TAP instance to a slightly bigger one by a gadget expansion at every node so that any feasible solution to the EDGE-LP on the original instance is feasible to the NODE-LP in the expanded instance. Moreover, EDGE-LP has extreme points which violate our conjecture above, motivating a deeper study of ODD-LP for future work.
    
\item In Section~\ref{sec:threecycle}, we provide a complete study of 3TAP in which every tree edge must be in a triangle in the final solution. Via a reduction from set cover, we show an $\Omega(\log n)$-inapproximability result and give a matching approximation algorithm. In the unweighted case, we show that any minimal solution gives a 4-approximation.

\end{enumerate}

Our approach is a top-down coloring algorithm on the scaled fractional solution where each color class is a feasible solution. In particular, $\frac32$ times the fractional solution is decomposed into a convex combination of integer solutions. This provides not only an approximation algorithm but also directly proves the integrality gaps for the corresponding covering LPs~\cite{carr2000randomized}. In addition, this technique of top-down coloring differs from all current $\frac32$-approximation algorithms on unweighted TAP and all current algorithms which achieve better than $2$-approximations for special cases of weighted TAP. Since our methods decompose scaled fractional solutions, they also have the potential to extend to give tight integrality gap proofs - we propose some ideas for doing this in Section~\ref{sec:obs}.

\section{Problem Structure}
\label{sec:structure}
In this section, we show that we can restrict our attention to only certain instances of weighted TAP. This structure restricts not only the structure of the links but also the structure of the tree itself.

\begin{theorem}
\label{thm:structure}
Any instance of weighted TAP $(T,c, L)$ can be reduced to a corresponding instance of weighted TAP $(T', c', L')$ of roughly the same size where the tree $T'$ is binary and all the leaves in $L'$ go between two leaves. In addition, every feasible solution to $(T,c,L)$ provides a feasible solution to $(T',c', L')$ of equal cost and vice versa.
\end{theorem}

\begin{figure}
\centering
\begin{tikzpicture}
\node[circle, fill=black,thick, inner sep=2pt, minimum size=0.1cm, label=left:$v_0$](v) at (0,2) {};
\node[circle, fill=black,thick, inner sep=2pt, minimum size=0.1cm, label=left:$v_1$](c1) at (-1,1) {};
\node[circle, fill=black,thick, inner sep=2pt, minimum size=0.1cm, label=left:$v_2$](c2) at (0,1) {};
\node[circle, fill=black,thick, inner sep=2pt, minimum size=0.1cm, label=left:$v_3$](c3) at (1,1) {};
\draw[thick] (0, 2.5) -- (v) -- (c1) -- (-1.5, .5);
\draw[thick] (c1) -- (-.5, .5);
\draw[thick] (v) -- (c2) -- (0, .5);
\draw[thick] (v) -- (c3) -- (.6, .5);
\draw[thick] (1,.5) -- (c3) -- (1.4, .5);

\node[circle, fill=black,thick, inner sep=2pt, minimum size=0.1cm, label=left:$v_0'$](v0a) at (7,4) {};
\node[circle, fill=black,thick, inner sep=2pt, minimum size=0.1cm, label=left:$v_0$](v0b) at (6,3) {};
\node[circle, fill=black,thick, inner sep=2pt, minimum size=0.1cm, label=left:$v_1'$](v1a) at (8,3) {};
\node[circle, fill=black,thick, inner sep=2pt, minimum size=0.1cm, label=left:$v_1$](v1b) at (7,2) {};
\node[circle, fill=black,thick, inner sep=2pt, minimum size=0.1cm, label=left:$v_2'$](v2a) at (9,2) {};
\node[circle, fill=black,thick, inner sep=2pt, minimum size=0.1cm, label=left:$v_2$](v2b) at (8,1) {};
\node[circle, fill=black,thick, inner sep=2pt, minimum size=0.1cm, label=left:$v_3'$](v3a) at (10,1) {};
\node[circle, fill=black,thick, inner sep=2pt, minimum size=0.1cm, label=left:$v_3$](v3b) at (9,0) {};
\node[circle, fill=black,thick, inner sep=2pt, minimum size=0.1cm, label=left:$v_4'$](vb) at (11, 0) {};

\draw[thick] (7,4.5)--(v0a)--(v1a)--(v2a)--(v3a)--(vb);
\draw[thick] (v0a)--(v0b);
\draw[thick] (v1a)--(v1b);
\draw[thick] (6.5,1.5)--(v1b) --(7.5, 1.5);
\draw[thick] (v2a)--(v2b) -- (8,.5);
\draw[thick] (v3a)--(v3b)--(9,-.5);
\draw[thick] (8.6, -.5)--(v3b) --(9.4, -.5);
\draw[thick, dashed] (v0b) to[out=-90,in=-135] (vb);
\end{tikzpicture}
\caption{An example of a $v$ node with three children before and after the transformation.}
\label{fig:bin}
\end{figure}
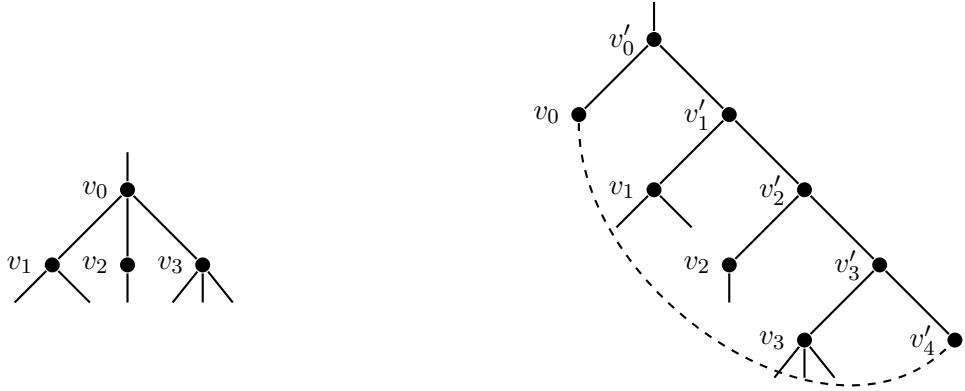

The construction is a local operation performed on all the nodes of $T$ in a top-down fashion. Let $v_0$ be a node in the tree with children $v_1, v_2, \dots v_k$ (if $v_0$ is a leaf then no operation will be done). Let $(T,c,L)$ be the initial tree, the transformation on $v$ will give us a new instance $(T_v, c_v, L_v)$. We will add dummy nodes $v_i'$ for $v$ and all its children and a dummy nodes $v_{k+1}'$ for $v$. We remove the edges $X=\{v_0v_i\}_i$ and add the edges $Y=\{v_iv_i'\}_i \cup \{v_i'v_{i+1}'\}_i$.  We leave all the existing links at their corresponding nodes. The only link we add is a link called $\ell_v$ from $v_0$ to $v_{k+1}'$ of cost $0$. The new instance has changed as follows:
\begin{align*}
V(T_v)&=V(T)\cup \{v_i'\}_i\\
E(T_v)&=E(T)-X + Y\\
L_v&=L\cup \{v_0v_{k+1}'\}
\end{align*}
Figure~\ref{fig:bin} gives an example of this transformation on a node with three children.

We will now show that performing this transformation on every non-leaf vertex of $T$ produces an instance of TAP with a binary tree and leaf-to-leaf links with corresponding feasible solutions to the original problem.

\begin{proof}
First we observe that this transformation adds nodes $v_0', v_1', \dots v_k'$ all of degree $3$, adds node $v_{k+1}'$ of degree 1, and node $v_0$ ends with degree $1$. The transformation also keeps the degree of $v_1, v_2\dots v_k$ unchanged. Once this transformation has been applied to all non-leaves of $T$ then the resulting tree $T'$ will have only nodes of degree $1$ and $3$; giving a binary tree as desired.

Now observe that every original node is a leaf in $T'$. The only links we added were $\ell_v$ which have the form $v_{k+1}'$ to $v_0$ where $v_{k+1}'$ is also a leaf under the transformation. The resulting set of links $L'$ is leaf-to-leaf.

We will now consider any feasible solution $A$ to $(T,c, L)$. Let $A'=A\cup\{\ell_v\}_{v\in V}$. The cost of $A$ and $A'$ are the same as we added only links $\ell_v$ which were given cost $0$. First observe that $\ell_v$ covers all the edges of the form $v_{i}'v_{i+1}'$ and $v_0'v_0$. Now consider an edge $v_i'v_i$ after the transformation. There is some link $\ell\in A$ which covers $v_0v_i$ in $T$ and now that same link must cover $v_i'v_i$ in $T'$. So, $A'$ is a valid solution to $(T', c', L')$ of the same cost.

Now let $A'$ be a feasible solution $(T', c', L')$ now consider there is a vertex $v_0$ which was not initially a leaf node in $T$. It must be the case that $A'$ contains $\ell_v$ as this is the only link in $L'$ which covers $v_k'v_{k+1}'$. So, let $A=A'-\{\ell_v\}_{v\in V}$. Now by the same argument as previously, as $A'$ is a feasible solution for $T'$ and the only edges in $T'$ not in $T$ are those covered by the $\ell_v$ then $A$ is a valid solution to $(T,c,L)$. Notice that $A$ and $A'$ have the same cost as we only removed links of cost $0$ from the solution.
\end{proof}

\section{Large Links}
\label{sec:large}
The main result of this section is the following theorem.

\begin{theorem}
\label{thm:largelinksTwo}
Given a solution $x$ to the EDGE-LP with $x_\ell \geq \alpha$ when $x_\ell >0$ and $m$ is the number of non-zero links then there exists integer solutions $x^1, x^2, \dots x^{2m}$ and $\lambda_1,\dots, \lambda_{2m}$ such that:
\[
\frac{2}{1+\alpha} x \leq \sum_{i=1}^{2m} \lambda_i x^i
\]
and this convex combination can be found in strongly polynomial time.
\end{theorem}

This gives an alternative proof of Cheriyan, Jord\'an and Ravi~\cite{cheriyan1999}. In particular, it gives a $\frac43$-approximation when we start with a fractional solution where all non-zero links have weight at least $\frac12$.

\subsection{Algorithm}
We will be working with a tree rooted at an arbitrary node, $r$. The least common ancestor (LCA) of a link, is the least common ancestor of its endpoints. We let $L_\ell$ ($R_\ell$) be the path in the tree from the LCA of $\ell$ to the left (right) endpoint of $\ell$ (this path could be empty).

Given a fractional solution, $x$ let $\alpha = \min_{\ell: x_\ell \neq 0} x_\ell$, and let $\beta = \frac{2}{1+\alpha}$. Let $k$ be the smallest integer such that $k\beta x$ is an even integer for all entries. In order to find our convex decomposition in the algorithm below, we will decompose $k \beta x$ into $k$ different color such that each color is a feasible tree augmentation.

The main idea of how the algorithm works is that it goes down the tree looking at links which have their LCA at the current node and colors all the copies of each link with different colors so as to help cover the edges as much as possible with new colors. This guarantees that the first $\alpha \beta k$ links (copies of one link) which are colored through an edge all get distinct colors. Afterward, we only guarantee that of the remaining links that cover an edge half of them give a new color to that edge.

\begin{algorithm}[h]
\KwData{$T$ a tree, $x$ LP solution, $\beta$ approximation factor, $k$ colors}
\KwResult{Decomposition of $k \beta x$ into $k$ different colors where each color is a feasible tree augmentation}
Make $k\beta x_\ell$ copies of each link $\ell$\;
\While{some link is not colored}{
$\ell$ has the highest LCA among uncolored links\;
	\While{not all copies of $\ell$ colored}{
	Color a copy of $\ell$ with the first color not present on $L_\ell$\;
	\If{all edges of $L_\ell$ are covered by all $k$ colors}{
	Color a copy of $\ell$ with any color not already on a copy of $\ell$\;
	}
	Color a copy of $\ell$ with the first color not present on $R_\ell$\;
	\If{all edges of $R_\ell$ are covered by all $k$ colors}{
	Color a copy of $\ell$ with any color not already on a copy of  $\ell$\;
	}
	}
}

 \caption{The simple coloring algorithm}
\label{alg:color}
\end{algorithm}

We will now show that this coloring does indeed give us a convex combination as desired.

\begin{theorem}
\label{thm:alg}
Algorithm~\ref{alg:color} guarantees that every edge is covered by a link in every one of the $k$ colors.
\end{theorem}

\begin{proof}
For a given $e$ without all $k$ colors, every time a link through $e$ receives a pair of colors, then one of those colors is new to $e$. Let us consider some link $\ell$ through $e$. Each inner while loop of the algorithm gives two colors to copies of $\ell$. One of the two paths $L_\ell, R_\ell$ must contain $e$; without loss of generality let $e\in L_\ell$. Consider the highest edge $f\in L_\ell$ without all $k$ colors. If $f$ is missing a color $c$, then $e$ must also be missing color $c$. We have only colored links whose LCA is above $f$, therefore any link with a color which covers $e$ must also cover $f$.  So, for each pair of colors chosen for a link through $e$, at least one of them is a new color for $e$. In other words, half of the time a link covering $e$ gets colored, it is a new color for $e$.

The first time a link through an edge $e$ is colored, then all its colors are distinct (unless $\beta x_\ell >1$). For a given link $\ell$, every time a color is picked for a copy of $\ell$ it has to be a color not on one of the edges $\ell$ covers or a color not on any copy of $\ell$. If $\beta x_\ell >1$, then we can color the copies of $\ell$ with all $k$ colors and all the edges which $\ell$ covers will be covered by all $k$ colors. In this case, $e$ would get all $k$ colors.

Thus the first time an edge has one of its links colored it receives at least $\alpha \beta k$ distinct colors. Combining this with the fact that every edge gets colors at rate $\frac{1}{2}$ subsequently, the total number of colors $e$ receives in this process is at least
\[
\alpha\beta k +\frac{1-\alpha}{2}{\beta} k = \frac{1+\alpha}{2}\beta k = k.
\]
\end{proof}

Now we will show how this implies Theorem~\ref{thm:largelinksTwo}.

\begin{proof}[Proof of Theorem~\ref{thm:largelinksTwo}]
By scaling our $x$ up by $k\beta$ we can write this scaled version as the sum of $k$ different feasible colors (integer solutions). This gives us that:
\[
k\beta x = \sum_{i=1}^k x^i
\]
where the $x^i$ are integer solutions. Dividing by $k$ gives the desired result.

Algorithm~\ref{alg:color} does not need to first multiply by $\beta k$ before being run. The algorithm can be run by just multiplying the solution by $\beta$. As the algorithm runs, it will keep track of a convex combination of integer partial solutions. In each while loop when a link $\ell$ is added, $\ell$ will be fully added to some integer partial solutions and added to a fraction of at most two partial integer solutions (one for $R_\ell$ and one for $L_\ell$). This creates at most two more integer partial solutions. The number of different integer solutions at the end can be bounded by $2m$ where $m$ is the number of non-zero links. This guarantees this algorithm can be run in strongly polynomial time.
\end{proof}

\subsection{Conjecture}

Theorem~\ref{thm:largelinksTwo} deals with the case when $x$ does not have fractional parts which are very small. In particular, the case where $\alpha = \frac{1}{3}$ gives a $\frac{3}{2}$ approximation with this algorithm. Another approach to this problem would be to iteratively round when a solution has a link with fractional value at least $\frac{2}{3}$ (See e.g., ~\cite{lrs}).

In particular, when a fractional solution has $x_\ell\geq \frac23$ we can immediately round up $x_\ell$ to $1$ and resolve the linear program with this added constraint. This approach combined with using our approximation when $x_\ell\geq \frac13$ for all $\ell$ would achieve a $\frac32$ approximation as every individual link gets rounded up by at most $\frac32$ and the cost of the residual LPs do not increase in the process.

By combining these two approaches one would be able to provide a $\frac{3}{2}$ approximation to weighted TAP. It would be very convenient if every fractional solution had one of these two properties: a link $\ell$ with $x_\ell\geq 2/3$, or $x_\ell\geq 1/3$ for all non-zero $x_\ell$. Unfortunately, there exists extreme points of the EDGE-LP which satisfy neither of these properties as shown in Section~\ref{sec:obs}. Therefore, we propose the following conjecture.
\begin{conjecture}
\label{con:onethird}
Every extreme point solution $x^*$ to the ODD-LP has one of the two properties: $x^*_\ell \geq 1/3$ for all non-zero $x^*_\ell$ or there is some $\ell$ such that $x^*_\ell \geq 2/3$.
\end{conjecture}

\section{Deficient Paths}
\label{sec:deficient}
In this section, we will start with a solution to NODE-LP and use the additional structure from the constraints~\ref{eq:node} to help us. Due to our previous observation as we will assume that the TAP instance is a binary tree with all links going from leaf-to-leaf. We will break edges into two groups depending on how much coverage they receive.

\begin{definition}
An edge $e\in T$ is considered \emph{deficient} if $x(\delta(e)) < 4/3$ and \emph{abundant} if $x(\delta(e)) \geq 4/3$.
\end{definition}
The deficient edges in a solution to the NODE-LP can not be too dense; this would violate the node constraint~\ref{eq:node}.
\begin{lemma}
The deficient edges form paths in $T$.
\end{lemma}
\begin{proof}
Suppose there was a node, $v$, adjacent to three deficient edges: $e_1, e_2, e_3$. By the node inequality~\ref{eq:node}, we know that:
\[x(\delta(e_1) \cup \delta(e_2) \cup \delta(e_3)) \geq 2\]
In this particular case, every link through $v$ goes through exactly $2$ of $e_1, e_2, e_3$. So, we have:
\[x(\delta(e_1) \cup \delta(e_2) \cup \delta(e_3))\leq \frac{1}{2}(x(\delta(e_1)) +x(\delta(e_2))+x(\delta(e_3))) <2 \]
This is a contradiction to the feasibility of $x$ for the NODE-LP. So, there is no node with three deficient edges in a solution to the NODE-LP, and the deficient edges form paths as desired.
\end{proof}

\subsection{A Top-Down Greedy 2-approximation and Ramifications}
In this section, we present a simple $2$-approximation which will be used to deal with the abundant edges in future cases. There are numerous $2$-approximations for TAP, but we will use a specific coloring one as it allows us to extend colorings.

Choose any vertex $r$ to be the root. Let $k$ be the smallest non-negative integer such that $kx_\ell$ is an integer for all links $\ell$. For this approach, we will multiply our fractional solution by $4k$ and then break it up into $2k$ integral solutions. The cost of the cheapest such solution will be at most $4k/2k=2$ times the cost of the original.

We will be using LCA($\ell$), $R_\ell$ and $L_\ell$ as defined in the previous section.

\begin{algorithm}
\KwIn{Tree $T$, root $r$, feasible solution $x$ to EDGE-LP, least common multiple $k$}
\KwOut{Breaks $4kx$ into $2k$ colors each of which is a solution}
\For {Links $\ell$} {
	Break the $4kx_\ell$ into $2kx_\ell$ copies of $R_\ell$ and $2kx(f)$ copies of $L_\ell$\;
}
\While{Not all $2k$ colors are solutions} {
	Let $e$ be the highest edge without all $2k$ colors\;
	Choose an uncolored link $\ell$ in $x(\delta(e))$\;
	Choose a color $c_i$ not on $e$\;
	Color $\ell$ with $c_i$\;
}
Transfer the colors of $L_\ell,R_\ell$ back to $\ell$\;
\caption{Greedily colors the links representing of the EDGE-LP top-down to give $2k$ solutions.}
\label{alg:greedy}
\end{algorithm}

The top-down algorithm is given in~\ref{alg:greedy}. The main idea is to double each link and use one copy to cover the left path from its lca and the other for its right path. In this sense, it is reminiscent of the approach of Frederickson and J\'aJ\'a~\cite{frederickson1981approximation} of splitting each cross link in the tree to two up links to devise a 2-approximation algorithm. The main idea of the coloring algorithm is to only supply colors to links that are missing at one of the edges it covers. Since the links are colored top down, this ensures that any color missing at an edge is also missing in all its descendant edges.

\begin{lemma}\label{lem:greedy}
The $2k$ colors returned by~\ref{alg:greedy} are valid solutions.
\end{lemma}
\begin{proof}
Consider any $e\in T$. When the algorithm starts, there are at least $4k$ links which cover $e$, because $x(\delta(e)) \geq 1$. After the transformation of the links, there are at least $2k$ edges which cover $e$.

As the algorithm progresses, the colors covering every edge are a subset of those covering its parent. Let $p$ be the parent of $e$. The first time we color a link through $e$ that is not through $p$, then we must have given $p$ all $2k$ colors already.

Every time a link through $e$ gets a color, it is because some edge $e'$ above $e$ was missing that color. By the above observation, the colors missing from $e'$ are also missing from $e$. Therefore, $e$ also got a new color. Hence, every time $e$ gets one of it's $2k$ links colored, it gets a new color.

Every edge is covered by all $2k$ colors, so every color is a solution as desired.
\end{proof}
The correctness of the algorithm implies that taking the cheapest color (in terms of total link cost) is a valid solution, leading to the following result.

\begin{corollary}
There is a greedy top-down coloring based $2$-approximation for TAP.
\end{corollary}

When none of the edges are deficient, then we can push this result even further.
\begin{corollary}
Given a solution $x$ to TAP with no deficient edges, $3kx$ can be decomposed into $2k$ feasible colors.
\end{corollary}
\begin{proof}
We can re-use Algorithm~\ref{alg:greedy}  and its proof. The only thing we have to change is that we break $3kx_\ell$ into two parts of size $\frac32kx_\ell$. Since all edges are abundant, $x(\delta(e)) \geq 4/3$ and so after the split, every edge $e$ has $2k$ links covering it.
\end{proof}

We will strengthen this further to allow us to finish off the abundant parts after we deal with the deficient parts of the tree in later proofs.
\begin{definition}
A rooted subtree is considered \emph{abundant} if all its edges are abundant.
\end{definition}
\begin{definition}
A partial coloring of $3kx$ causes a \emph{conflict} if there is an edge $e$ which is covered by three links of the same color $c$ and $e$ does not yet have all $2k$ colors covering it. A partial coloring is considered \emph{conflict-free} if it causes no conflicts.
\end{definition}

In particular, we show that given the start of the coloring we can finish it off if we didn't do too much wrong.
\begin{theorem}
\label{thm:abundSubtree}
Given a partial conflict-free coloring of $3kx$ on some links through the root of an abundant subtree, it can be extended to cover all the edges in the subtree.
\end{theorem}
\begin{proof}
We can simply start the greedy algorithm at the root and finish every edge off.
Every edge in the subtree has at least $4k$ links covering it originally. For every color $c$ through an edge $e$, if there are two copies of that color, then $e$ can pretend that one link with color $c$ was originally given to the side of the cycle formed by the link not through $e$.
\end{proof}

\subsection{One Deficient Path}
We now extend the greedy coloring algorithm to show that if the deficient edges only form one path in the tree, then there exists a $\frac32$ approximation.

Consider we have a solution $x$ to NODE-LP with only one deficient path $P$; let $u_1, u_2, \dots u_j$ be the deficient path. We will deal with this case by first coloring some links such that every edge in $P$ gets all $2k$ colors. Then we will split up all the uncolored links that go through $P$. We root the tree at $u_1$, then we will use Theorem~\ref{thm:abundSubtree} to finish all the abundant subtrees.

\begin{algorithm}
\KwIn{Tree $T$ with one deficient path $P=u_1u_2\dots u_j$, feasible solution $x$ to NODE-LP, least common multiple $k$}
\KwOut{Breaks $3kx$ into $2k$ colors which cover $P$ and is conflict-free at all abundant edges}
\For{ $u_iu_{i+1}$ an edge in $P$} {
	\For{ Color $c$ not covering $u_iu_{i+1}$} {
		Pick an uncolored link, $\ell$, through $u_iu_{i+1}$\;
		Color $\ell$ with color $c$\;
	}
}
\While{There is some $u_iu_{i+1}$ with at least three links of color $c$}{
	Let the three links through $u_iu_{i+1}$ of color $c$ be $\ell_1,\ell_2,\ell_3$\;
	With respect to the edges in $P$, let $\ell_1$ cover only a subset of the edges covered by $\ell_2$ and $\ell_3$\; (At least one such labeling exists by making $\ell_2$ and $\ell_3$ the two links with the furthest coverage in the two sides of $u_iu_{i+1}$ respectively)
	Uncolor $\ell_1$\;
}
\caption{Greedily colors the links to give $2k$ solutions that cover the path. Also, it avoids overcoloring the links through abundant edges and results in a conflict-free coloring.}
\label{alg:onePath}
\end{algorithm}

\begin{theorem}
Algorithm~\ref{alg:onePath} provides all $2k$ colors to all the edges in $P$ and is conflict-free.
\end{theorem}
\begin{proof}
As the algorithm progresses, an edge in $P$ either has all $2k$ colors or all of the colored links covering it are distinct colors. If a edge $e$ in $P$ were to receive a duplicate color $c$ before all $2k$ colors, then there were two edges $e_1, e_2$ that both needed $c$. Without loss of generality let $e_2$ lie between $e$ and $e_1$. When $e_1$ takes color $c$ on a link that also covers $e$, then the link must cover $e_2$. This contradicts that $e_2$ would need color $c$. Therefore, every edge in $P$ gets all $2k$ colors in the first part of the iteration.

The clean-up phase does not remove any colors from edges in $P$. A link $\ell$ of color $c$ that becomes uncolored, is uncolored only if all the edges in $P$ it covers have color $c$ from other links. So, the clean-up phase never removes any colors from edges of $P$.

The clean-up phase guarantees that for all the rooted subtree, the coloring induces at most two links of that color and therefore causes no conflicts.
\end{proof}

Now by combining the previous theorem, and Theorem~\ref{thm:abundSubtree} there is a $\frac32$ approximation when there is  only one deficient path.

\begin{corollary}
Given a solution $x$ for NODE-LP on tree $T$ which induces at most one deficient path, there is an integral solution of cost at most $\frac32$ the cost of $x$.
\end{corollary}

\subsection{Two Deficient Paths}
We've shown how to deal with a single deficient path. To extend our approach to two deficient paths, we need to deal with the abundant path of tree edges which connects the two deficient paths in the binary tree . The goal is to color links to cover the two deficient paths and the abundant path between them, and maintain that every remaining abundant subtree doesn't receive too many copies of each color. We will prove the following in this section:

\begin{theorem}
\label{thm:twoDefPaths}
Given a solution $x$ on tree $T$ which induces at most two deficient paths, there is an integral solution of cost at most $\frac32$ the cost of $x$.
\end{theorem}

To deal with the deficient edges and the abundant path between them, we will have to examine the structure of links near deficient edges.

\begin{lemma}
\label{lem:def}
Given a solution $x$ to the NODE-LP and a deficient edge $e=uv$ where $u$ is an internal node, then the total weight of links through $u$ but not $e$ is at least $\frac23$.
\end{lemma}

\begin{proof}
Let the neighbors of $u$ be $v,w_1,w_2$. The triangle constraint on edges $uv, uw_1, uw_2$ says the total weight of the links that cover these edges is at least $2$. All the links that cover $uv, uw_1,$ or $uw_2$ also go through $u$. So, the total weight of links through $u$ is at least $2$. The weight of links through the deficient edge $uv$ is less than $\frac43$ by definition; this gives the total weight of links through $u$ but not through $uv$ is at least $\frac23$ as desired.
\end{proof}

Consider $x$ is a solution to NODE-LP with two deficient paths; let $P_1, P_2$ be the deficient paths and let $Q=q_1q_2\dots q_j$ be the abundant path which connects them. In order to deal with this case, we will first color the links that form the intersections of $Q$ with $P_1$ and $P_2$. Then, we will finish coloring $Q$. Lastly, we will color $P_1, P_2$. Throughout the whole process we will also guarantee that the coloring is conflict free, i.e., every abundant edge not in $Q$ gets all $2k$ colors or has at most two copies of every color. In order to deal with this coloring, we will have to treat links which cover all of $Q$ differently.

\begin{definition}
A link $\ell$ is considered to be a \emph{long link} if $\ell$ covers all the edges of $Q$. A link $\ell$ which is not a long link, is considered to be a \emph{short link}.
\end{definition}

Let $e_1, e_2$ be the edges of $P_1$ adjacent to $Q$. Let $e_3, e_4$ be the edges of $P_2$ adjacent to $Q$. If $P_1$ or $P_2$ only has one edge adjacent to $Q$ this simplifies the case greatly. We will address this case last. Algorithm~\ref{alg:pair} finds a coloring starting at the ends of $Q$.

\begin{algorithm}
\KwIn{Tree $T$, feasible solution $x$ to NODE-LP, least common multiple $k$, adjacent deficient edges $e_1, e_2$ in $P_1$ between which the abundant path $Q$ to the other deficient path $P_2$ originates}
\KwOut{Pairs up $2k$ links covering $e_1$ and $2k$ links covering $e_2$ in a way that avoids overusing any edge not on $Q$}
\For{ $i=1$ to $2k$}{
	Let $f_i$ be the link through $e_1$ which covers the $i$th most number of edges of $Q$\;
	Let $g_i$ be the link through $e_2$ which covers the $i$th most number of edges of $Q$\;
}
Pair up the edges such that $f_i$ gets paired with $g_{2k+1-i}$\;
\For {$e$ not on $Q$} {
	\While {$f_i,g_j$ are paired and both cover $e$} {
		Choose a second pair $f_{i'}, g_{j'}$ where neither covers $e$ \;
		\If{No such pair $f_{i'}, g_{j'}$ exists}{
			Break from the While loop\;
		}
		Change the pairing so that $f_i, g_{j'}$ are paired and $f_{i'}, g_j$ are paired\;
	}
}
\caption{Takes two adjacent deficient edges $e_1, e_2$ and pairs up $2k$ of the $3kx$ links covering $e_1$ and $e_2$ in a way that avoids any edge not in $Q$ being covered by both links in a pair unless it is covered by all pairs.}
\label{alg:pair}
\end{algorithm}

\begin{lemma}
Algorithm~\ref{alg:pair} finds a pairing such that every edge not on $Q$ is either covered by all $2k$ pairs or is only covered by at most one link in every pair.
\end{lemma}
\begin{proof}
We only need to check that a swap can not cause an edge $e'$ not on $Q$ to have more pairs with both links covering $e$. Consider that some swap occurred because $e$ was covered by both $f_i, g_j$ but not covered by $f_{i'}, g_{j'}$. If $e'$ is covered by both links in both pairs $f_i, g_{j'}$ and $f_{i'}, g_j$ then it was covered by both links in the pairs before the swap. So, consider that $e'$ is covered by both links $f_i, g_{j'}$ after the swap but was not covered by both links in either of the pairs before the swap. So, neither of $f_{i'}, g_j$ cover $e'$.

Right now we have that $f_i$ covers $e_1, e', e$ but not $e_2$; $g_j$ covers $e_2, e$ but not $e_1, e'$; $f_{i'}$ covers only $e_1$; $g_{j'}$ covers $e_2, e'$ and not $e_1, e$.

Consider the tree $T'$ which is all the edges of $T$ contracted except for $e, e', e_1, e_2$. The edges $e_1,e_2$ were adjacent in $T$ so they are still adjacent in $T'$. $e_1$ can't separate $e_2$ and $e$ as $g_j$ covers $e$ and $e_2$. Similarly, $e_2$ can't separate $e_1$ and $e$ because of $f_i$. Likewise, $e_1$ can't separate $e_2$ and $e'$ because of $g_{j'}$, and $e_2$ can't separate $e_1$ and $e'$ because of $f_i$. This leaves only the three possibilities shown in Figure~\ref{fig:possibilities}.

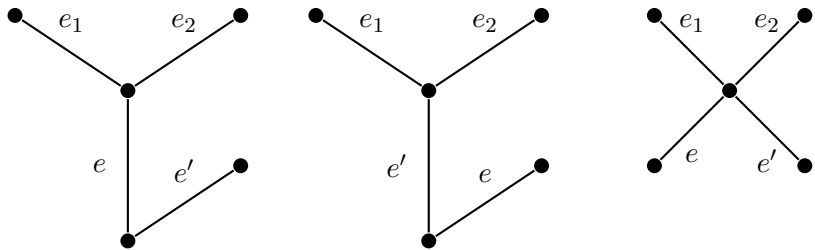
\begin{figure}
\centering
\begin{tikzpicture}
\node[circle, fill=black,thick, inner sep=2pt, minimum size=0.1cm](a0) at (0,0) {};
\node[circle, fill=black,thick, inner sep=2pt, minimum size=0.1cm](b0) at (-1.5,1) {};
\node[circle, fill=black,thick, inner sep=2pt, minimum size=0.1cm](c0) at (1.5,1) {};
\node[circle, fill=black,thick, inner sep=2pt, minimum size=0.1cm](d0) at (0,-2) {};
\node[circle, fill=black,thick, inner sep=2pt, minimum size=0.1cm](e0) at (1.5,-1) {};
\draw[thick] (a0) -- (b0) node [midway, label= above:$e_1$] {};
\draw[thick] (a0) -- (c0) node [midway, label= above:$e_2$] {};
\draw[thick] (a0) -- (d0) node [midway, label= left:$e$] {};
\draw[thick] (d0) -- (e0) node [midway, label= above:$e'$] {};
\node[circle, fill=black,thick, inner sep=2pt, minimum size=0.1cm](a1) at (4,0) {};
\node[circle, fill=black,thick, inner sep=2pt, minimum size=0.1cm](b1) at (2.5,1) {};
\node[circle, fill=black,thick, inner sep=2pt, minimum size=0.1cm](c1) at (5.5,1) {};
\node[circle, fill=black,thick, inner sep=2pt, minimum size=0.1cm](d1) at (4,-2) {};
\node[circle, fill=black,thick, inner sep=2pt, minimum size=0.1cm](e1) at (5.5,-1) {};
\draw[thick] (a1) -- (b1) node [midway, label= above:$e_1$] {};
\draw[thick] (a1) -- (c1) node [midway, label= above:$e_2$] {};
\draw[thick] (a1) -- (d1) node [midway, label= left:$e'$] {};
\draw[thick] (d1) -- (e1) node [midway, label= above:$e$] {};
\node[circle, fill=black,thick, inner sep=2pt, minimum size=0.1cm](a2) at (8,0) {};
\node[circle, fill=black,thick, inner sep=2pt, minimum size=0.1cm](b2) at (7,1) {};
\node[circle, fill=black,thick, inner sep=2pt, minimum size=0.1cm](c2) at (9,1) {};
\node[circle, fill=black,thick, inner sep=2pt, minimum size=0.1cm](d2) at (7,-1) {};
\node[circle, fill=black,thick, inner sep=2pt, minimum size=0.1cm](e2) at (9,-1) {};
\draw[thick] (a2) -- (b2) node [midway, label= above:$e_1$] {};
\draw[thick] (a2) -- (c2) node [midway, label= above:$e_2$] {};
\draw[thick] (a2) -- (d2) node [midway, label= below:$e$] {};
\draw[thick] (a2) -- (e2) node [midway, label= below:$e'$] {};
\end{tikzpicture}
\caption{The three possible configurations of the edges $e_1, e_2, e, e'$}
\label{fig:possibilities}
\end{figure}

In the first case of Figure~\ref{fig:possibilities}, then $g_{j'}$ must cover $e$ as it covers $e_2$ and $e'$. This is a contradiction as $g_{j'}$ does not cover $e$. In the second case of Figure~\ref{fig:possibilities}, $g_j$ would cover $e'$ as it covers $e_2$ and $e$. This is a contradiction as $g_j$ does not cover $e'$. Due to $f_i$ covering $e', e_1, e$, it must be the case that $e_1, e, e'$ are all on a path; this removes the third case in Figure~\ref{fig:possibilities}.

Therefore, the swaps never increase the number of pairs which both cover an edge. When the algorithm ends, every pair covers an edge, or that edge is covered in every pair.
\end{proof}
Note that this swapping algorithm would work regardless of which $2k$ links through $e_1$, and through $e_2$ were chosen and  how we initially paired them.

Now we can use Algorithm~\ref{alg:pair} to start a coloring on either side of the abundant path $Q$. We now need to coordinate the pairings and then finish coloring the abundant path. The first thing to observe, we can swap any two long links as long as the two edges of $e_1,e_2,e_3,e_4$ they cover are the same. This will never create any edges which have two links from the same pairing but not a link from every pairing. To coordinate the pairings on either side of $Q$ then we need to deal with the long links. We consider them in three mutually exclusive and collectively exhaustive cases: the first case where there are at most $2k$ long links, the second when there are more than $2k$ long links, and finally, when we have $2k$ long links covering one of $e_1,e_2,e_3, e_4$.

\paragraph{At most $2k$ long links.} In this case, we first observe that the initial pairing will use all the long links, and none of the long links are paired up with each other in Algorithm~\ref{alg:pair}. In addition, a swap will always be initiated by a pair of short links $f_i, g_j$; if it were a pair of a long link and a short link, $f_i, g_j$ would have no common edges outside of $Q$. So, any swap involves at most one long link. Therefore, a swap never creates a pair of two long links.

Use Algorithm~\ref{alg:pair} to pair up $2k$ links for $e_1,e_2$ and $2k$ links for $e_3, e_4$. Every long link is used in a pairing for $e_1,e_2$ and a pairing for $e_3, e_4$ and no pairing has two long links. We use Algorithm~\ref{alg:extend} to color the links to cover $Q$.

\begin{algorithm}
\KwIn{Tree $T$, solution $x$ to NODE-LP, least common multiple $k$, adjacent deficient edges $e_1, e_2$, $2k$ pairs of links through $e_1, e_2$, adjacent deficient edges $e_3, e_4$, $2k$ pairs of links through $e_3, e_4$}
\KwOut{Colors the pairs and some other links in a conflict-free way to cover $Q$}
\For{ $f_i,g_i$ one of the $2k$ pairs covering $e_1,e_2$} {
	Color $f_i,g_j$ with an unused color $c_i$\;
	\If {$f_i$ or $g_i$ is a long link} {
		WLOG let $f_i, h_i$ be a pair of $e_3, e_4$\;
		Color $h_i$ with $c_i$\;
	}
}
\For{ $i=1$ to $j-1$} {
	\While{$q_iq_{i+1}$ doesn't have a color $c$ } {
	Let $f$ be the uncolored link covering $q_iq_{i+1}$ and the most number of edges on $Q$ after $q_iq_{i+1}$\;
	Color $f$ with color $c$\;
	\If{$f$ is in a pair for $e_3, e_4$} {
		Let $f,g$ be the pair\;
		Color $g$ with color $c$\;
		}
	}
}
\caption{Takes the pairs from Algorithm~\ref{alg:pair} and extends them to cover all of $Q$, and be conflict free}
\label{alg:extend}
\end{algorithm}

\begin{lemma}
The coloring produced by Algorithm~\ref{alg:extend} is conflict-free, when there are at most $2k$ long links.
\end{lemma}
\begin{proof}
Suppose for a contradiction that there was some edge $e$ covered by 3 links of a color $c$ but not covered by all $2k$ colors. This edge $e$ is not on $Q$ (as all edges of $Q$ are covered by all $2k$ colors). Let $q_i$ be the closest vertex of $Q$ to $e$ in $T$. Without loss of generality, two of the links of color $c$ ($f_1,f_2$) through $e$ also go through $q_{i-1}$ (the other possibility is that there is two links of color $c$ through $q_{i+1}$. If both $f_1,f_2$ were part of a pair from $e_1, e_2$, then by Algorithm~\ref{alg:pair} then it would be the case that $e$ has all $2k$ colors covering it. So, without loss of generality $f_1$ was added to extend the coloring of color $c$ along $Q$ when $f_2$ was already colored $c$. If $f_1$ covers a subset of the edges of $Q$ that $f_2$ does, then $f_1$ never would have been colored with color $c$ by Algorithm~\ref{alg:extend}. If $f_2$ covers a subset of the edges of $Q$ that $f_1$ does, then $f_1$ would have been chosen and colored before $f_2$. Both of these are a contradiction to the correct running of Algorithm~\ref{alg:extend}.
\end{proof}

This shows when there are at most $2k$ long links, we can form a conflict-free partial coloring that covers all of $Q, e_1,e_2,e_3, e_4$.

\paragraph{More than $2k$ long links} In this case, we don't have to worry about covering $Q$ as every pair will have a long link. However we will have to be a little more careful with the coordination of the two pairings. In this case, consider there are $2k+c$ long links. We will first show we can create $c$ pairs that use two long links, and then just create the remaining $2k-c$ pairs with one long link each.

\begin{lemma}
If there are $2k+c$ long links, and no edge of $e_1,e_2,e_3,e_4$ has more than $2k$ long links, then there exists $c$ pairs of long links such that each of these $c$ pairs cover $e_1,e_2,e_3, e_4$.
\end{lemma}
\begin{proof}
Let $c_{ij}$ be the number of links that cover $e_i$ and $e_j$. A pair that covers all of $e_1,e_2,e_3,e_4$ consists of a long link through $e_1,e_3$ and a long link through $e_2, e_4$ or a long link through $e_1,e_4$ and a long link through $e_2, e_3$.   Let $c_A=\min(c_{13}, c_{24}), c_B= \min(c_{14}, c_{23})$. If there were not $c$ pairs that covered all of $e_1, e_2, e_3,e_4$ then $c_A+c_B < c$. Without loss of generality, $c_A=c_{13}, c_B=c_{14}$ so the number of long links covering $e_1$ is less than $c$. But then, the total number of long links covering $e_2$ is more than $2k$. This is a contradiction. Therefore, we can make $c$ pairs such that each of these $c$ pairs cover $e_1,e_2,e_3, e_4$.
\end{proof}

For this case, we simply create these $c$ pairs of long links indicated in the lemma above as the first $c$ colors. Next, we can use the remaining $2k-c$ long links as the start to fill the remaining colors using Algorithm~\ref{alg:pair}. Now, if $e_i$ isn't covered by a color, pick an uncolored link $f$ that doesn't go through any other $e_j$ and color $f$ with the remaining color. This creates no conflicts. Each color consisting of two long links covers only the edges of $Q$ twice and everywhere else once. Each color with only one long link only has two short links, and the color covers each of $e_1,e_2,e_3,e_4$ exactly once; each edge can only have at most two links of this color covering it.

\paragraph{More than $2k$ long links cover a single $e_i$}
In this case, without loss of generality $e_1$ has $2k$ long links. In this case, we simply start by creating and coloring the $2k$ pairs as in Algorithm~\ref{alg:pair} for $e_3, e_4$. Each of these pairs will cover all of $Q$ and $e_1,e_3,e_4$. For each color which $e_2$ is missing, let $f$ be an uncolored link covering $e_2$ but not $e_1$. Color $f$ with the color that $e_2$ is missing. This again keeps the coloring conflict-free by the same reasoning as before. This gives a partial coloring that covers all of $Q, e_1,e_2,e_3,e_4$.

\paragraph{Finishing the deficient paths}
In all the cases above, we colored the links to create a conflict-free partial coloring that covers all of $Q$ and $e_1,e_2,e_3,e_4$. To finish the deficient paths, we will use Lemma~\ref{lem:def} to observe there are $2k$ uncolored links crossing every deficient edge. Consider starting at $e_1$ and moving away from $Q$ along $P_1$. Let $e$ be the deficient edge; there are $2k$ links through $e$ not through $e_1$. For each color $e$ is missing, choose one of these uncolored links for this color. Repeat this moving away from $e_1$ along $P_1$. We can do this same process along $P_1$ moving away from $e_2$. This guarantees that we only add colors to the abundant subtrees hanging off of $P_1$ (and not to abundant subtrees hanging off of $Q$ or $P_2$). In addition, we can add each color at most twice to every subtree. We can repeat this process similarly with $e_3,e_4$ on $P_2$. This gives a conflict-free partial coloring which gives all the edges of $P_1, P_2, Q$ all $2k$ colors. By Theorem~\ref{thm:abundSubtree}, this coloring can be completed and there is a $\frac32$ approximation when there are only two deficient paths.

\paragraph{End of $P_1, P_2$}
If $P_1$ or $P_2$ have only one edge adjacent to $Q$ the entire process above can be done simply without creating a pairing on such a side. Consider $P_1$ only has one edge $e_1$ adjacent to $Q$; instead of $2k$ pairs created for $P_1$ just take the $2k$ links through $e_1$ which go the furthest down $Q$. This replaces Algorithm~\ref{alg:pair} for $P_1$ and then we proceed according to the case we are in.
\\ \\
This proves Theorem~\ref{thm:twoDefPaths} as in all the cases we have provided a $2k$ coloring of $3kx$ such that every color is a feasible solution. This coloring method could potentially be extended to the  general case of multiple deficient paths, with the key difficulty being that there are links that cover segments of potentially several such deficient paths, and their colorings must be somehow globally coordinated. We have no counter examples to the success of such a potential approach even though we have no candidate algorithm that might complete the job for all feasible solutions of NODE-LP. 
\section{Comparing the linear programs}
\label{sec:obs}

In section~\ref{sec:large}, we proposed a conjecture that every extreme point solution to the ODD-LP has a link which gets $x_{\ell}\geq 2/3$ or all the non-zero $x_\ell$ are at least $1/3$. This conjecture does not hold for extreme points of the EDGE-LP. An example is given in Figure~\ref{fig:counterexample}.

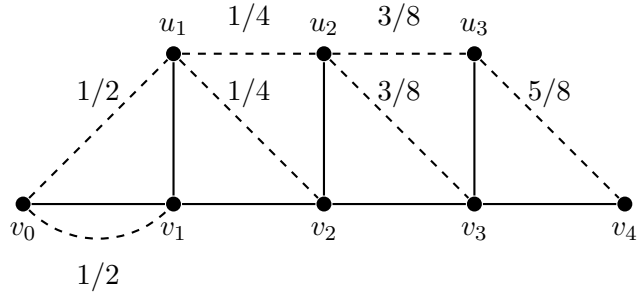
\begin{figure}[h!]
\centering
\begin{tikzpicture}
\node[circle, fill=black,thick, inner sep=2pt, minimum size=0.1cm, label=below:$v_0$](v0) at (0,0) {};
\node[circle, fill=black,thick, inner sep=2pt, minimum size=0.1cm, label=below:$v_1$](v1) at (2,0) {};
\node[circle, fill=black,thick, inner sep=2pt, minimum size=0.1cm, label=below:$v_2$](v2) at (4,0) {};
\node[circle, fill=black,thick, inner sep=2pt, minimum size=0.1cm, label=below:$v_3$](v3) at (6,0) {};
\node[circle, fill=black,thick, inner sep=2pt, minimum size=0.1cm, label=below:$v_4$](v4) at (8,0) {};
\node[circle, fill=black,thick, inner sep=2pt, minimum size=0.1cm, label=above:$u_1$](u1) at (2,2) {};
\node[circle, fill=black,thick, inner sep=2pt, minimum size=0.1cm, label=above:$u_2$](u2) at (4,2) {};
\node[circle, fill=black,thick, inner sep=2pt, minimum size=0.1cm, label=above:$u_3$](u3) at (6,2) {};
\draw[thick] (v0) -- (v1) -- (v2) -- (v3) -- (v4);
\draw[thick] (v1) -- (u1);
\draw[thick] (v2) -- (u2);
\draw[thick] (v3) -- (u3);
\draw[thick, dashed] (v0) -- (u1) node [midway, label= above:$1/2$] {};
\draw[thick, dashed] (u1) -- (u2) node [midway, label= above:$1/4$] {};
\draw[thick, dashed] (u2) -- (u3) node [midway, label= above:$3/8$] {};
\draw[thick, dashed] (u3) -- (v4) node [midway, label= above:$5/8$] {};
\draw[thick, dashed] (v0) to[out=-45,in=-135] (v1);
\node [label= below:$1/2$] at (1, -.5) {};
\draw[thick, dashed] (u1) -- (v2) node [midway, label= above:$1/4$] {};
\draw[thick, dashed] (u2) -- (v3) node [midway, label= above:$3/8$] {};
\end{tikzpicture}
\caption{An example of an extreme point of EDGE-LP which doesn't fulfill the conjecture.}
\label{fig:counterexample}
\end{figure}

We believe the conjecture to be true for extreme points of ODD-LP. We have performed exhaustive searches on all extreme points of ODD-LP to verify it on small trees (binary trees with at most 12 nodes). In addition, Fiorini et. al~\cite{fiorini2017frac} showed that in the special case of no in-links then ODD-LP is integer. Given a rooted tree $T$ with root $r$ then an \emph{in-link} is a link $\ell$ where $\ell$ does not go through $r$ and  $\ell$ does not go from one node to it's ancestor. 
This indicates that the ODD-LP potentially has more structure than the EDGE-LP that might be exploited to prove the conjecture.

In this paper, we use the structure given from the NODE-LP (e.g. Lemma~\ref{lem:def}) to prove theorem~\ref{thm:twoDefPaths}. While the NODE-LP does add some constraints to the EDGE-LP, the NODE-LP is not much stronger than the EDGE-LP as we show in the next observation. In particular, we can transform any TAP instance to a slightly bigger one by a gadget expansion at every node so that any feasible solution to the EDGE-LP on the original instance is feasible to the NODE-LP in the expanded instance. This shows that if we were to do the gadget expansion for any input, the NODE-LP constraints alone (without the other ODD=LP constraints) will not add any strength to the resulting solutions. 

\begin{lemma}
After a transformation of any TAP instance that leaves the solutions unchanged, the integrality gap for the NODE-LP  on the transformed instance is the same as the integrality gap for the EDGE-LP in the original instance.
\end{lemma}
\begin{proof}
In order to show this result, we will show that for every binary tree $T$ we can transform it to another binary tree $T'$ such that every feasible solution of the EDGE-LP for $T$ corresponds to a feasible solution of the  NODE-LP for $T'$.

\begin{figure}
	\centering
	\begin{subfigure}{.45\textwidth}
		\centering
		\begin{tikzpicture}
			\tikzstyle{every node}=[draw,circle,fill=white,minimum size=4pt, inner sep=0pt]
			\draw (0,0) node (v) [label = right: $v$] {} -- (60: 2cm) node (1) [label = right: $n_1$] {};
			\draw (v) -- (180: 2cm) node (2) [label = left: $n_2$] {};
			\draw (v) -- (300: 2cm) node (3) [label = right: $n_3$] {};
		\end{tikzpicture}
		\caption{The original internal node $v$ and it's neighborhood}
		\label{fig:orig}
	\end{subfigure}
	\begin{subfigure}{.45\textwidth}
		\centering
		\begin{tikzpicture}
			\tikzstyle{every node}=[draw,circle,fill=white,minimum size=4pt, inner sep=0pt]
			\draw (0,0) node (v) [label = right: $v$] {} -- (60: 1.5cm) node (a1) [label = right: $a_1$] {}-- (60: 3cm) node (n1) [label = right: $n_1$] {};
			\draw (a1) -- (90: 1.5cm) node (b1) [label = left: $b_1$] {};
			\draw (v) -- (180: 1.5cm) node (a2) [label = above: $a_2$] {}-- (180: 3cm) node (n2) [label = left: $n_2$] {};
			\draw (a2) -- (210: 1.5cm) node (b2) [label = below: $b_2$] {};
			\draw (v) -- (300: 1.5cm) node (a3) [label = left: $a_3$] {}-- (300: 3cm) node (n3) [label = right: $n_3$] {};
			\draw (a3) -- (330: 1.5cm) node (b3) [label = above: $b_3$] {};
			\draw[dashed] (b1) -- (b2) -- (b3);
		\end{tikzpicture}
		\caption{The gadget which gets placed around $v$. The dotted links are added and given weight zero. No other links are added adjacent to $b_1,b_2,b_3$}
		\label{fig:new}
	\end{subfigure}
	\caption{The left figure shows the original node, and the right figure shows the node after the gadget transformation has been applied to it}
\end{figure}
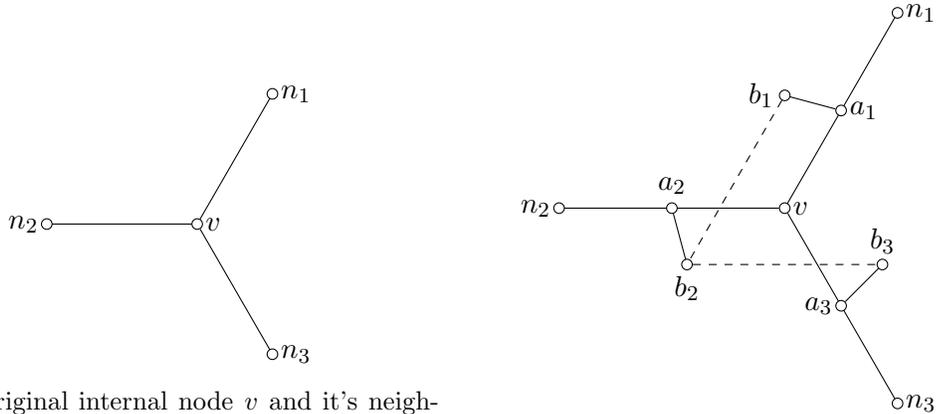
Consider an original node $v\in T$ with neighbors $n_1,n_2,n_3$ as shown in Figure~\ref{fig:orig}. We will transform $v$ as shown into the structure given in Figure~\ref{fig:new}.

Any feasible solution to the EDGE-LP on $T$ can be made a feasible solution to the NODE-LP on $T'$ by adding all the zero-cost links with value 1 in the solution. These zero cost links guarantee all new tree edges that we added are covered. Every degree 3 node, $u$ has at least 2 weight of links covering its neighbors. 

Consider any feasible solution to the NODE-LP on $T'$. For every node $v$the solution must choose links $b_1b_2$ and $b_2b_3$ with value 1 for otherwise $a_1b_1$ and $a_3b_3$ will not be covered. To get a feasible solution to the EDGE-LP on $T$ we simply remove the links $b_1b_2, b_2b_3$ around every node.

Neither of these transformations change the cost, so we have proven our result as desired.
This gadget forces $b_1,b_3$ to be covered fully by the dotted links. These ensure that the node constraints around $v,a_1,a_2,a_3$ are all satisfied as long as the edge constraints are satisfied.


\end{proof}

While the NODE-LP is not any stronger than the EDGE-LP, we believe that the ODD-LP is stronger than both of these LPs. As noted before, the ODD-LP may add key constraints which would allow us to extend the top-down deficient path coloring beyond just one or two deficient paths.

\section{Three-Cycle TAP}
\label{sec:threecycle}
\subsection{Weighted Version}
In this section, we will consider the weighted version of 3TAP where the weights on the links, $c_f$, can take on any value. We first present an $O(\log n)$ approximation algorithm, and then we present a matching lower bound of $\Omega(\log n)$, where $n$ is the number of nodes in the tree.

\begin{theorem}
There is a $O(\log n)$ approximation algorithm for weighted 3TAP on $n$ nodes.
\end{theorem}
\begin{proof}
Consider any feasible solution $A$ to 3TAP, such that $T\cup A$ has every tree edge in a 3-cycle. For a vertex $v$, let $\delta(v)$ be the edges of $T\cup A$ adjacent to $v$.

To turn this problem into a set cover problem, we let the edges of $E(T)$ be the elements. For any subset of edges adjacent to a vertex $v$, $S_v$, we construct a set with cost $c(S_v)$ which covers the edges induced by the endpoints of the edges in $S_v$ (except for those edges adjacent to $v$).

By doubling each edge in the feasible solution, then we can decompose the entire solution into these stars. So, given a solution to 3TAP of total cost $C$, the corresponding set cover has a solution of at most $2C$. Given any solution to the set cover problem, then we simply add all the edges specified by the stars to the tree (with maybe some duplicates when edges are added from both endpoints' stars). Thus, any solution of cost $C$ to the set cover gives a solution to the original 3TAP solution of cost between $C/2$ and $C$. Therefore the optimal solutions to these two problems are within a factor of two of each other.

It is well known that minimum-cost set cover with $n$ elements has an $O(\log n)$ approximation as long as the densest set (that has the maximum ratio of newly covered elements divided by the cost of the set) can be found in polynomial time. For a fixed vertex $v$, we can find the maximum density star centered at $v$ as follows. Due to a result by Goldberg, one can find the maximum density subgraph; $S\subset V$ which minimizes $\frac{|E(S)|}{c(S)}$ in polynomial time~\cite{maxden}.
For the given center $v$, by setting the cost of another vertex $u$ to be $c(uv)$, we can use the maximum density subgraph algorithm to find the maximum density star from $v$, where the edges in the subgraph are the tree edges covered in triangles by the corresponding star. By repeating this for every choice of center vertex $v$, we can find the maximum density star in polynomial time.  This gives the maximum density set for the set cover problem in polynomial time. Then we can use the greedy algorithm for set cover to get an $O(\log n)$ approximation for 3TAP.
\end{proof}

Notice that in the above algorithm we used no properties of the original graph $T$. This algorithm will in fact work for any graph $T$ where the goal is to augment such that every edge is in a 3-cycle.

\begin{corollary}
The problem of finding a minimum cost augmentation of any graph $G$ where every edge must be in a 3-cycle has an $O(\log n)$ approximation.
\end{corollary}

The above approximation is tight as the weighted 3TAP problem captures set-cover exactly. We will now show the matching lower bound.

\begin{theorem}
3TAP does not have a $\Omega(\log n)$-approximation unless $NP \subseteq P$.
\end{theorem}

\begin{figure}
\begin{center}
\begin{tikzpicture}
\tikzset{circle node/.style={draw,circle,fill=white,minimum size=4pt,inner sep=0pt}}
\node at (1,4) (d) [circle node, label = above:$s$] {};
\node at (0,3) (e) [circle node, label = left:$r$] {};
\node at (1,3) (S1) [circle node, label = above:$S_1$]{};
\node at (2,3) (S2) [circle node, label = above:$S_2$]{};
\node at (3,3) {$\cdots$};
\node at (4,3) (Sk) [circle node, label = above:$S_k$]{};
\node at (1,2) (c) [circle node, label = left: $t$] {};
\node at (0,1) (e1) [circle node, label = below:$e_1$] {};
\node at (1,1) (e2) [circle node, label = below:$e_2$] {};
\node at (2,1)   {$\cdots$};
\node at (3,1) (en) [circle node, label = below:$e_n$] {};
\draw (d)--(e)--(S1)--(S2);
\draw (e)--(c)--(e1);
\draw (c)--(e2);
\draw (c)--(en);
\end{tikzpicture}
\end{center}
\caption{The 3TAP instance created from a set cover instance.}
\label{fig:gadget}
\end{figure}
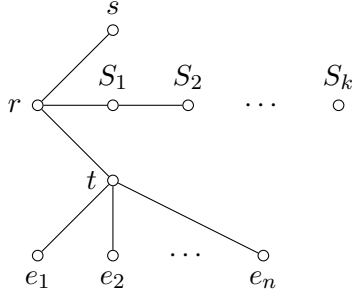

\begin{proof}
Consider an instance of set cover with sets $S_1, S_2, \dots, S_k$ and elements $e_1,e_2,\dots e_n$ and cost function $c$ on the sets. We will have our tree be as shown in Figure~\ref{fig:gadget}. The vertex set is
\[
\{r,s,t\}\cup \{S_i\}_{i=1}^k \cup \{e_j\}_{j=1}^n
\]
with the following costs on the links:
\begin{itemize}
\item Links from $s$ to vertices $\{r,t\} \cup \{S_i\}_{i=1}^k$
\item Links from $t$ to $S_i$ have cost $c(S_i)$
\item If $e_j \in S_i$ then the link from $e_j$ to $S_i$ has cost 0
\item All remaining links have cost $1+\sum_{i=1}^k c(S_i)$. Call the set of these remaining edges $L$.
\end{itemize}

In any optimal solution, we will not use any edges from $L$ as taking all the edges not in that set have smaller cost and give a feasible solution. The zero edges from $s$ allow every edge except for the $te_j$ edges to be in a three-cycle and they have cost 0. Now the only way to have an edge $te_j$ in a three cycle is for $tS_i$ and $e_jS_i$ to be used for some $S_i$ such that $e_j\in S_i$. So, the non-zero edges bought correspond to sets being chosen.

Consider any feasible solution to set cover $S_{i_1}, \dots S_{i_k}$, this can be turned into a feasible solution to 3TAP of the same cost. All the zero cost edges in addition to the $tS_{i_\ell}$ edges form a feasible solution. We know all the edges in the tree except for the $te_j$ edges are in a three cycle with zero cost edges. Consider any $j\in [n]$. There is some $S_{i_t}$ that contains $j$. The edge $te_j$ is then in a three cycle with $tS_{i_\ell}$ and $S_{i_\ell}e_j$. Hence, every feasible solution to the set cover instance gives a feasible solution of the same cost to the 3TAP instance.

Consider any feasible solution to our 3TAP instance. If the 3TAP solution contains an edge from $L$ then the solution has weight at least $1+\sum_{i=1}^k c(S_i)$, by taking all the sets $S_i$ we get a feasible solution to the set cover instance of less cost. Now consider there are no edges from $L$ in the feasible solution for the 3TAP instance. Let $tS_{i_1}, \dots tS_{i_t}$ be the non-zero edges in the solution. Therefore $S_{i_1}, \dots S_{i_t}$ is a feasible solution to the set cover instance. Consider any element $e_j$. The edge $te_j$ must be in some three cycle with $tS_{i_\ell}$ and $S_{i_\ell}$ therefore, $S_{i_\ell}$ contains $e_j$ and is a set in our solution to set cover. Therefore every feasible solution to 3TAP has a corresponding solution of set cover with the same or smaller cost.

Any feasible solution to set cover gives a solution to 3TAP of the same cost. Any feasible solution to 3TAP, gives a feasible solution to set cover of the same or smaller cost. Therefore, by the hardness of set cover, it is impossible to approximate three-cycle TAP to within a $\Omega(\log n)$ factor unless $NP \subseteq P$~\cite{arora2003improved}.
\end{proof}

{\bf Remark:}
Suppose we were given an empty initial graph to augment and wish to find a minimum-cost two-edge-connected spanning subgraph where every edge is in a triangle, it is not hard to adapt the above hardness: We give all edges in the tree zero cost. By further subdividing the path of set nodes  $S_1, S_2, \dots, S_k$ to add new dummy nodes between every pair of set nodes, we can ensure that every element node $e_j$ is covered only by triangles containing edge $(t,e_j)$. This requires that the other edges in the cycle are of the form $(t,S_i), (S_i, e_j)$ for some set $S_i$ containing the element $e_j$.

\subsection{Unweighted Version}
While weighted 3TAP has many similarities to set cover, the unweighted version admits a constant approximation unlike set cover. Here we consider the case that every non-tree edge has cost either $1$ or infinity, and every tree edge is present (and has cost $0$). This 4-approximation comes from lower bounding the cost of every feasible solution to unweighted 3TAP.

\begin{lemma}
Every feasible unweighted 3TAP solution has cost at least $\frac{n-1}{2}$.
\end{lemma}
\begin{proof}
Consider any solution $S$. Duplicate all the links of $S$ and  edges $T$ and forming stars around every vertex consisting of the edges adjacent to it. Call the star around $v$, $S_v$.  This doubles the cost of the solution, but now we can see that every tree edge is covered by some star.  At every vertex, we can further decompose $S_v$ into $S_v^1, \dots S_v^{\ell_v}$ such that we get stars that cover different connected components of the tree and every star contains at most one tree edge.

Now consider any star $S_v^i$. If $S_v^i$ has $x$ links, then the number of tree edges it can cover with 3-cycles is at most $x$. So, in the doubled instance of $S$ there must be at least $n-1$ edges. Every link is in at most $2$ stars; there must be at least $\frac{n-1}{2}$ edges in any feasible solution.
\end{proof}

\begin{corollary}
Unweighted 3TAP has a 4-approximation.
\end{corollary}
\begin{proof}
We can get a $4$ approximation by simply taking any minimal feasible solution. For every edge $ab$, pick a $v$ such that $av,bv$ both have cost $0$ or $1$. If no such vertex exists, then no feasible solution exists. Otherwise, the algorithm chooses at most $2(n-1)$ links. This gives a $4$ approximation as desired.
\end{proof}

\section{Conclusions}
We have introduced a new top down coloring method that gives a strict improvement over existing 2-approximation algorithms for weighted TAP, with better improvements for larger minimum values in the LP. Our methods give constructive convex combinations into feasible solutions and when coupled with the strengthened ODD-LP for the problem have much potential to settle the integrality gap for this fundamental network design problem. We also settled the approximation complexity of the special case when all edges in the final solution must be in triangles -- the extensions to short constant-length cycles in place of triangles is immediate. We hope our new algorithms will provide a stepping stone to settling the integrality gap for weighted TAP.

\bibliographystyle{plain}
\bibliography{refs}

\appendix
\section{Odd Constraints}
\label{app:odd}

In this section, we provide a new proof of correctness of the ODD-LP.

\begin{lemma}
The constraints in ODD-LP are valid for any integer solution to TAP.
\end{lemma}

\begin{proof}[Proof by Robert Carr]
Consider an odd set of vertices $S$. By adding together the edge constraints for $\delta(S) \cap T$ we get:
\[
\sum_{e\in \delta(S)\cap T} x(\delta(e)) \geq |\delta(S)\cap T|
\]
Now we can add any non-negative terms to the left hand side and still remain feasible. Therefore

\[
x(\delta(S))+\sum_{e\in \delta(S)\cap T} x(\delta(e)) \geq |\delta(S)\cap T|
\]
is also feasible. Now consider any link $\ell$. If $x_\ell$ appears an even number of times in $\sum_{e\in \delta(S)\cap T} x(\delta(e))$ then $\ell$ is not in $\delta(S)$. Similarly, if $x_\ell$ appears an odd number of times in $\sum_{e\in \delta(S)\cap T} x(\delta(e))$ then $\ell$ is in $\delta(S)$. So, the coefficient of every $x_\ell$ on the left hand side of this expression is even. In particular, for any integer solution the left hand side is even and the right hand side is odd. Therefore, we can strengthen the right hand side by increasing it by one, and the resulting constraint will still be feasible for any integer solution. The constraint
\[
x(\delta(S))+\sum_{e\in \delta(S)\cap T} x(\delta(e)) \geq |\delta(S)\cap T|+1
\]
is valid for any integer solution to TAP as desired.
\end{proof}

\end{document}